\documentclass[11pt]{article}

\usepackage[numbers]{natbib}
\usepackage{amsmath}
\usepackage{amssymb}
\usepackage{amsthm}
\usepackage{xcolor}
\usepackage{array}
\usepackage{booktabs}
\usepackage{tabularx}
\usepackage{hyperref}
\usepackage{algorithm}
\usepackage{algpseudocode}
\usepackage{enumitem}\usepackage[margin=1in]{geometry}

\hypersetup{colorlinks=true,citecolor=blue,linkcolor=blue,hypertexnames=false}
\allowdisplaybreaks

\newcommand{\paperTitle}{Alphabet-Preserving Lifting for the Log-Rank Conjecture}
\newcommand{\paperAuthor}{Zhao Song\thanks{\texttt{magic.linuxkde@gmail.com}. The author would like to thank Omri Weinstein for useful discussions.}}

\DeclareMathOperator{\rank}{rank}
\newcommand{\E}{\mathbb E}
\newcommand{\R}{\mathbb R}

\theoremstyle{plain}
\newtheorem{theorem}{Theorem}[section]
\newtheorem{lemma}[theorem]{Lemma}
\newtheorem{definition}[theorem]{Definition}
\newtheorem{proposition}[theorem]{Proposition}
\newtheorem{corollary}[theorem]{Corollary}
\theoremstyle{definition}

\begin{document}

\date{\today}
\title{\paperTitle}
\author{\paperAuthor}
\maketitle

\begin{abstract}
For a Boolean communication matrix $M$, let $D(M)$ denote its deterministic
communication complexity and let $r(M):=\rank_{\R}(M)$.  The log-rank
conjecture asks whether $D(M)$ is polynomial in $\log r(M)$.  The best known
general upper bound, due to Sudakov and Tomon~\cite{st25}, is
$D(M)=O(\sqrt{r(M)})$.

On the lower-bound side, G{\"o}{\"o}s, Pitassi, and Watson constructed explicit
matrices satisfying $D(M)=\Omega((\log r(M))^2/(\log\log r(M))^2)$~\cite{gpw18}.
We improve the lower bound to
$D(M)=\Omega((\log r(M))^2/\log\log r(M))$.

Our construction revisits their pointer function over its original non-Boolean
alphabet and lifts it with an alphabet-valued Index gadget, via the multicolor
simulation theorem stated by Roughgarden and Weinstein~\cite{rw16}.  Compared with the quantitatively explicit GPW bound,
the alphabet-preserving lift removes one factor of $\log\log r$.  We also give a
self-contained proof of the multicolor simulation theorem in the parameter
regime required by the construction.

\end{abstract}

%%% This file is the structure for main body content
%%% This file should only contain use \input{xxx}.
%%% TeX files for body contents should be named as:
%%% 01_xxxx.tex
%%% 02_xxxx.tex
%%% ...
%%% 49_xxxx.tex

\section{Introduction}
\label{sec:introduction}

Communication complexity asks how much two separated parties must communicate to evaluate a function of their joint input.  For a Boolean function
$
f\colon X\times Y\to\{0,1\},
$
its communication matrix is $M_f:=(f(x,y))_{x\in X,y\in Y}$.  We write $D(f):=D(M_f)$ for its deterministic two-party communication complexity and
$
r:=r(M_f):=\rank_{\R}(M_f)
$
for its rank over the reals.  All logarithms in this paper are to base two.  The elementary rank bounds are
$
  \log r(M_f)\le D(f)\le r(M_f)+1.
$ 
Indeed, a cost-$c$ deterministic protocol partitions $M_f$ into at most $2^c$ monochromatic rectangles, and the indicator matrices of its $1$-leaves express $M_f$ as a sum of at most $2^c$ rank-one matrices.  Conversely, a Boolean matrix of rank $r$ has at most $2^r$ distinct rows, so Alice can identify her row using at most $r$ bits, after which Bob announces the function value using one further bit.  

The log-rank conjecture of Lov\'asz and Saks~\cite{ls88,ls93} asserts that the first inequality is polynomially tight: there is a universal constant $C$ such that every Boolean matrix $M$ satisfies
$
  D(M)\le (\log r(M))^C.
$ 
The formulation is insensitive to the common choice between $\{0,1\}$-matrices and sign matrices: if $A:=2M-J$, where $J$ is the all-ones matrix, then $|\rank(A)-\rank(M)|\le 1$.  A useful way to summarize the worst case is
$
  \mathcal C(r):=\sup\{D(M):M\text{ is Boolean and }\rank_{\R}(M)\le r\}.
$
Thus the conjecture asks whether $\mathcal C(r)=(\log r)^{O(1)}$.

\paragraph{Upper bounds}

The central structural question is whether low rank forces a large monochromatic rectangle.  Let $\alpha(s)$ be such that every rank-at-most-$s$ Boolean matrix contains a monochromatic rectangle occupying at least a $2^{-\alpha(s)}$ fraction of all entries.  Nisan and Wigderson~\cite{nw95} showed how to convert such a rectangle theorem into a complete protocol; in a convenient dyadic form their reduction gives
\begin{equation}
  \mathcal C(r)=O((\log r)^2+
  \sum_{i=0}^{\lceil\log r\rceil}
  \alpha(\lceil r/2^i\rceil)).
  \label{eq:nw_reduction}
\end{equation}
This reduction, rather than a new numerical record at the time, became the organizing principle for subsequent upper bounds.

Early progress came through graph-theoretic versions of the conjecture.  Work of Kotlov and Lov\'asz~\cite{kl96} controlled the number of distinct neighborhoods in a low-rank graph, and Kotlov~\cite{k97} obtained the constant-factor improvement $\mathcal C(r)\le \log(4/3)r+O(1)$, where $\log(4/3)\approx 0.415$.  This remained the best unconditional general estimate before the first genuinely sublinear bound.  An additive-combinatorial route was developed by Ben-Sasson, Lovett, and Ron-Zewi~\cite{blr14}: assuming the characteristic-two Polynomial Freiman--Ruzsa conjecture, they proved $\mathcal C(r)=O(r/\log r)$.  The required additive-combinatorial conjecture is now a theorem, following the resolution of Marton's conjecture by Gowers, Green, Manners, and Tao~\cite{ggmt25}; the resulting bound is nevertheless superseded by the stronger unconditional results below.

A decisive change came from discrepancy.  Building on the factorization-norm and discrepancy theory of Linial, Mendelson, Schechtman, and Shraibman~\cite{lmss07}, Lovett~\cite{l16} proved that every rank-$r$ Boolean matrix has a monochromatic rectangle of relative area $2^{-O(\sqrt r\log r)}$.  Combining this with Eq.~\eqref{eq:nw_reduction} yields $\mathcal C(r)=O(\sqrt r\log r)$, the first unconditional $o(r)$ bound.

The current record is due to Sudakov and Tomon~\cite{st25}.  If an $m\times n$ Boolean matrix has rank at most $r$ and density $p\le 1/2$, they establish the sparsity-sensitive discrepancy estimate $\operatorname{disc}(M)\ge c mn\min\{p,\sqrt p/\sqrt r\}$.  A density-decrement iteration then produces an all-zero or all-one submatrix with at least $m2^{-O(\sqrt r)}$ rows and $n2^{-O(\sqrt r)}$ columns.
Consequently $\alpha(r)=O(\sqrt r)$, and Eq.~\eqref{eq:nw_reduction} gives
\begin{equation}
  {\mathcal C(r)=O(\sqrt r)}.
  \label{eq:current_upper}
\end{equation}

\paragraph{Lower bounds}

The trivial rank argument gives $D(M)\ge\log r(M)$, and improving this separation has its own substantial history.  The graph formulation asks how large the chromatic number of a graph can be compared with the real rank of its adjacency matrix.  Alon and Seymour~\cite{as89} first disproved the most optimistic rank-coloring inequality by a finite example.  Razborov~\cite{r92} then constructed an asymptotic family with $\chi(G)=\Omega(\rank(A_G)^{4/3})$, which translates into only a constant-factor improvement over the elementary logarithmic communication lower bound.  Raz and Spieker~\cite{rs95} obtained the first superlinear-in-$\log r$ separation; in a rank-only statement, their construction gives $D(M)=\Omega(\log r\log\log\log r)$.

Nisan and Wigderson~\cite{nw95} gave an explicit family for which $D(M)=\Omega((\log r)^{\log_2 3})$, where $\log_2 3\approx 1.585$.  They also recorded Kushilevitz's unpublished six-variable base construction, which improves the exponent to $\log_3 6\approx 1.631$.
O'Donnell et al.~\cite{ostwz14} subsequently realized this construction in the XOR/Fourier-sparsity setting.  These examples are important evidence that any valid log-rank exponent must exceed one, but they remain far below the best upper bound.

Before the present work, the strongest general separation with an explicitly computed polyloglog loss was due to G\"o\"os, Pitassi, and Watson~\cite{gpw18}.  Their lifted pointer-function construction yields, for an explicit family,
\begin{equation}
  D(M)=\Omega({(\log r)^2}/{(\log\log r)^2})
  .
  \label{eq:current_lower}
\end{equation}
The polylogarithmic loss is worth displaying: statements of the result are often abbreviated as a near-quadratic lower bound.  Their work also produced a nearly quadratic separation between deterministic communication and one-sided rectangle partition number. A separation of the same near-quadratic order also follows from the construction of Balodis, Ben-David, G\"o\"os, Jain, and Kothari~\cite{bbgjk23} based on unambiguous DNFs; Hambardzumyan, Lovett, and Shirley~\cite{hls25} record it, with the polyloglog factors suppressed, as the largest known separation. The polyloglog loss is not computed in~\cite{bbgjk23}; in Section~\ref{sec:bbgjk} we trace their chain as written and show that it yields explicit matrices satisfying $D(M)=\Omega((\log r)^2/(\log\log r)^{7})$. Eq.~\eqref{eq:current_lower} therefore remains the sharpest previously known bound.

Combining Eq.~\eqref{eq:current_upper} and Eq.~\eqref{eq:current_lower}, the best previously known general bounds were therefore
\begin{equation*}
 {\Omega}((\log r)^2 / (\log \log r)^2 )
  \;\le\; \mathcal C(r)
  \;\le\; O(\sqrt r).
\end{equation*}
The lower bound is polylogarithmic in $r$, while the upper bound is exponential in the log-rank.  In particular, no unconditional $o(\sqrt r)$ upper bound is currently known, and no superpolynomial-in-$\log r$ lower bound is known.

\paragraph{Equivalent formulations}

Let $\chi_{\mathrm{part}}(M)$ denote the minimum number of disjoint monochromatic rectangles that partition all entries of $M$.  Protocol leaves give
$
  \log\chi_{\mathrm{part}}(M)\le D(M),
$
and Aho, Ullman, and Yannakakis~\cite{auy83} proved the converse simulation
$
  D(M)=O((\log\chi_{\mathrm{part}}(M))^2).
$
Hence the log-rank conjecture is equivalently a quasipolynomial relation between real rank and monochromatic rectangle partition number.  It is important here to distinguish this two-sided partition number from the one-sided quantity $p(M)$, the minimum number of disjoint all-one rectangles partitioning the $1$-entries; one always has
$
  r(M)\le p(M)\le 2^{D(M)}.
$

Several recent works have sharpened the structural picture without improving Eq.~\eqref{eq:current_upper}.  Singer and Sudan~\cite{ss22} developed equivalent incidence-geometric and parallel-partition formulations.  Hambardzumyan, Lovett, and Shirley~\cite{hls25} introduced the signed rectangle rank
$
  \operatorname{srr}(M):=\min \{t:M=\sum_{i=1}^t\varepsilon_iR_i,
  \ \varepsilon_i\in\{-1,1\},\ R_i\text{ a rectangle matrix} \}
$
and proved $\operatorname{srr}(M)=O(r(M)\log r(M))$.  They show that log-rank is equivalent to converting such a short signed decomposition into a positive/disjoint rectangle decomposition with only quasipolynomial blowup.  Hunter, Milojevi\'c, Sudakov, and Tomon~\cite{hmst24} extended the large-submatrix phenomenon to low-rank nonnegative and bounded-alphabet matrices and resolved a related disjoint-pairs conjecture.  Balla, Hambardzumyan, and Tomon~\cite{bht26} obtained an inverse theorem for bounded $\gamma_2$ norm, clarifying both the reach and the limitations of discrepancy/factorization-norm methods.  Contemporary graph formulations and Rank--Ramsey questions are developed further by Beniamini, Linial, and Shraibman~\cite{bls26}. Table~\ref{tab:log_rank_bounds} summarizes the general bounds.

\begin{table}[H]
\centering
\caption{Major general results on the deterministic log-rank problem. Upper bounds are universal, except for the explicitly marked conjectural row. In the lower panel, the 1982 row is universal, the 1989 row is a finite counterexample, and the remaining rows are asymptotic constructions.}
\label{tab:log_rank_bounds}
\begingroup
\normalsize
\setlength{\tabcolsep}{3.5pt}
\renewcommand{\arraystretch}{1.12}
\begin{tabularx}{\linewidth}{@{}
  >{\centering\arraybackslash}p{0.065\linewidth}
  >{\raggedright\arraybackslash}X
  >{\centering\arraybackslash}p{0.20\linewidth}
  >{\raggedright\arraybackslash}p{0.305\linewidth}
@{}}
\toprule
\multicolumn{4}{@{}l}{\textbf{Upper bounds}} \\
\addlinespace[1pt]
\textbf{Year} & \textbf{Authors} & \textbf{Refs} & \textbf{Bound} \\
\midrule
1982 & Mehlhorn and Schmidt & \cite{ms82} & $\le r+1$ \\
1988 & Lov\'asz and Saks & \cite{ls88,ls93} & Conjectured $(\log r)^{O(1)}$ \\
1995 & Nisan and Wigderson & \cite{nw95} & $O(\log^2r+\sum_i\alpha(r/2^i))$ \\
1996 & Kotlov and Lov\'asz & \cite{kl96} & $\le r/2+O(1)$ \\
1997 & Kotlov & \cite{k97} & $\le \log(4/3)r+O(1)$ \\
2014 & Ben-Sasson, Lovett, and Ron-Zewi & \cite{blr14} & $O(r/\log r)$ (conditional on PFR; unconditional via~\cite{ggmt25}) \\
2016 & Lovett & \cite{l16} & $O(\sqrt r\log r)$ \\
2025 & Sudakov and Tomon & \cite{st25} & $O(\sqrt r)$ \\
\midrule
\addlinespace[5pt]
\multicolumn{4}{@{}l}{\textbf{Lower bounds}} \\
\addlinespace[1pt]
\textbf{Year} & \textbf{Authors} & \textbf{Refs} & \textbf{Bound} \\
\midrule
1982 & Mehlhorn and Schmidt & \cite{ms82} & $\ge \log r$ \\
1989 & Alon and Seymour & \cite{as89} & $\chi(G)=32$, $\rank(A_G)=29$ \\
1992 & Razborov & \cite{r92} & $\ge (4/3)\log r-O(1)$ \\
1995 & Raz and Spieker & \cite{rs95} & $\Omega(\log r\log\log\log r)$ \\
1995 & Nisan and Wigderson & \cite{nw95} & $\Omega((\log r)^{\log_2 3})$ \\
1995 & Kushilevitz & \cite{nw95} & $\Omega((\log r)^{\log_3 6})$ \\
2018 & G\"o\"os, Pitassi, and Watson & \cite{gpw18} & $\Omega({(\log r)^2}/{(\log\log r)^2})$ \\
2023 & Balodis, Ben-David, G\"o\"os, Jain, and Kothari & \cite{bbgjk23} & $\Omega({(\log r)^2}/{(\log\log r)^{7}})$; see Section~\ref{sec:bbgjk} \\
2026 & This paper & Theorem~\ref{thm:main_lower_bound} & $\Omega({(\log r)^2}/{\log\log r})$ \\
\bottomrule
\end{tabularx}
\endgroup
\end{table}

\subsection{Our results}
We state our main result as follows.
\begin{theorem}[Main result, informal version of Theorem~\ref{thm:route_two_main}]
\label{thm:main_lower_bound}
There exists an explicit infinite family of total Boolean communication matrices
$(M_s)_{s\ge 2}$ such that, writing
$r_s:=\rank_{\R}(M_s)$,
\begin{equation*}
 D(M_s)
 =
 \Omega({(\log r_s)^2}/{\log\log r_s}).
\end{equation*}
\end{theorem}

Theorem~\ref{thm:main_lower_bound} also determines the worst case over all
matrices of bounded rank.

\begin{corollary}
\label{cor:main_cr}
For all sufficiently large $r$,
\begin{equation*}
 \mathcal C(r)
 =
 \Omega({(\log r)^2}/{\log\log r}).
\end{equation*}
\end{corollary}

\begin{proof}
Let $r$ be sufficiently large and let $s$ be maximal with $r_s\le r$; this is
well defined because $r_s\to\infty$. Since $\rank_{\R}(M_s)=r_s\le r$, the
definition of $\mathcal C$ gives
$\mathcal C(r)\ge D(M_s)=\Omega(s^2\log s)$, while
$\log r<\log r_{s+1}=O(s\log s)$ by
Theorem~\ref{thm:route_two_main}-\ref{part:route_two_rank_upper}. Eliminating
$s$ as in the proof of
Theorem~\ref{thm:route_two_main}-\ref{part:route_two_final} gives the claim.
\end{proof}

The same family, analysed in the same way, improves two further bounds of
G{\"o}{\"o}s, Pitassi, and Watson~\cite{gpw18}. The first concerns the
one-sided partition number $p(M)$ introduced above.

\begin{theorem}[Partition number, informal version of Corollary~\ref{cor:route_two_partition}]
\label{thm:main_partition}
The family of Theorem~\ref{thm:main_lower_bound} satisfies
\begin{equation*}
 D(M_s)
 =
 \Omega({(\log p(M_s))^2}/{\log\log p(M_s)}).
\end{equation*}
\end{theorem}

Yannakakis~\cite{y91} proved the best known general upper bound
$D(M)=O((\log p(M))^2)$ for every Boolean matrix $M$.

Since $r(M)\le p(M)$, Theorem~\ref{thm:main_partition} implies
Theorem~\ref{thm:main_lower_bound}. The second bound concerns the clique
versus independent set problem. For a graph $G$, let $\operatorname{CIS}_G$
denote the communication problem in which Alice receives a clique $x$ of $G$,
Bob receives an independent set $y$ of $G$, and the output is
$|x\cap y|\in\{0,1\}$.

\begin{theorem}[Clique versus independent set, informal version of Corollary~\ref{cor:route_two_cis}]
\label{thm:main_cis}
There exists an explicit infinite family of graphs $(G_s)_{s\ge 2}$, where
$G_s$ has $n_s$ vertices, such that
\begin{equation*}
 D(\operatorname{CIS}_{G_s})
 =
 \Omega({(\log n_s)^2}/{\log\log n_s}).
\end{equation*}
\end{theorem}

The corresponding bounds of~\cite{gpw18} carry a loss of
$(\log\log p(M))^2$ and of $(\log\log n)^2$, in place of the single
logarithmic factors above. Every $n$-vertex graph admits a clique
versus independent set protocol of cost $O((\log n)^2)$~\cite{y91}, and hence
Theorem~\ref{thm:main_cis} is optimal up to that single factor.

\paragraph{Notation.}
For a positive integer $n$, let $[n]:=\{1,\ldots,n\}$. For a Boolean communication problem $F$ with communication matrix $M_F$, write $\mathsf{P}^{\mathsf{cc}}(F):=D(M_F)$. For a function $f\colon\Sigma^n\to\{0,1\}$, let $\mathsf{P}^{\mathsf{dt}}(f)$ denote its deterministic decision-tree complexity, where each query reveals one symbol, and let $\mathsf{UP}^{\mathsf{dt}}(f)$ denote the minimum width of an unambiguous family of $1$-certificates.

\section{Proof overview}\label{sec:proof_overview}

\subsection{Summary of GPW} G{\"o}{\"o}s, Pitassi, and Watson~\cite{gpw18} begin with a pointer function $f:\Sigma^n\to\{0,1\}$, where $n=s^2$, satisfying $\mathsf{P}^{\mathsf{dt}}(f)=\Omega(n)$ and $\mathsf{UP}^{\mathsf{dt}}(f)=O(\sqrt n)$. They Booleanize each symbol using $\ell=\Theta(\log n)$ bits and an error-correcting-code decoder, obtaining $\mathsf{P}^{\mathsf{dt}}=\Omega(n\ell)$ and $\mathsf{UP}^{\mathsf{dt}}=O(\sqrt n\ell)$. Their Theorem~3 then lifts by Boolean Index, contributing another factor $\ell$, so $\mathsf{P}^{\mathsf{cc}}=\Omega(n\ell^2)$ and $L:=\log r=O(\sqrt n\ell^2)$. Since $x\mapsto x^2/(\log x)^2$ is increasing for large $x$, the upper bound on $L$ implies $L^2/(\log L)^2=O(n\ell^2)$, and hence $\mathsf{P}^{\mathsf{cc}}=\Omega(L^2/(\log L)^2)=\Omega((\log r)^2/(\log\log r)^2)$.

\subsection{Our approach} We retain the pointer function over its original alphabet. By Lemma~\ref{lem:route_two_pointer_parameters}, with $n=s^2$ it has $\mathsf{P}^{\mathsf{dt}}=\Omega(n)$ and $\mathsf{UP}^{\mathsf{dt}}=O(\sqrt n)$. After padding to $P=4n$ and setting $K:=P^{1000}$, Theorem~\ref{thm:route_two_lifting} applies the alphabet-valued Index gadget directly and gives $\mathsf{P}^{\mathsf{cc}}=\Omega(n\log n)$. Lemma~\ref{lem:route_two_up} converts the symbol certificates into a $1$-rectangle partition and yields $L:=\log r=O(\sqrt n\log n)$. Since $x\mapsto x^2/\log x$ is increasing for large $x$, this implies $L^2/\log L=O(n\log n)$, and therefore $\mathsf{P}^{\mathsf{cc}}=\Omega(L^2/\log L)=\Omega((\log r)^2/\log\log r)$.

\subsection{GPW lift vs alphabet-valued lift}
The improvement comes from preserving the input alphabet of the outer function, not from a stronger communication-lifting factor.  Theorem~3 of G{\"o}{\"o}s, Pitassi, and Watson~\cite{gpw18} takes an outer function with Boolean input alphabet.  Since one cell of the GPW pointer function is a symbol from an alphabet of size polynomial in $n$, their construction first Booleanizes every cell by an error-correcting-code decoder and then applies the Boolean Index lift.  By contrast, Theorem~5.5 of Roughgarden and Weinstein~\cite{rw16} permits the original alphabet-valued input to be lifted directly; Theorem~\ref{thm:route_two_multicolor} gives the specialized statement used here.  Both simulation theorems contribute a factor of order $\log n$ per outer query.  The difference is therefore the Booleanization step that precedes the GPW lift.

To track this distinction, let $n:=s^2$, let $\ell:=\lceil\log n\rceil$, and write $L:=\log r$.  In the GPW construction, Booleanization changes the symbol-level estimates into $\mathsf{P}^{\mathsf{dt}}=\Omega(n\ell)$ and $\mathsf{UP}^{\mathsf{dt}}=O(\sqrt n\ell)$; the subsequent Boolean lift contributes another factor $\ell$, giving communication lower bound $\Omega(n\ell^2)$ and $L=O(\sqrt n\ell^2)$~\cite{gpw18}.  In our construction, Lemma~\ref{lem:route_two_pointer_parameters} is used at the symbol level.  Theorem~\ref{thm:route_two_lifting} and Theorem~\ref{thm:route_two_multicolor} then give communication lower bound $\Omega(n\ell)$ and $L=O(\sqrt n\ell)$, without Booleanizing the input alphabet.

\begin{table}[ht]
\centering
\caption{Parameter comparison between GPW Theorem~3~\cite{gpw18} after Booleanization and the direct alphabet-valued lift supplied by RW Theorem~5.5~\cite{rw16}.  Both lifting theorems contribute a factor of order $\ell$; the extra factor in the GPW row is already present because each outer symbol was replaced by $\Theta(\ell)$ bits.  Here $n=s^2$, $\ell=\Theta(\log n)$, and $L=\log r$.}
\label{tab:gpw_alphabet_comparison}
\begin{tabularx}{\textwidth}{@{}>{\raggedright\arraybackslash}X*{3}{>{\centering\arraybackslash}X}@{}}
\toprule
Construction & CC lower bound & Upper bound on $L$ & After eliminating $n$ \\
\midrule
GPW Theorem~3 after Booleanization
& $\Omega(n\ell^2)$
& $O(\sqrt n\ell^2)$
& $\Omega(L^2/\ell^2)$ \\
RW Theorem~5.5, direct symbols
& $\Omega(n\ell)$
& $O(\sqrt n\ell)$
& $\Omega(L^2/\ell)$ \\
\bottomrule
\end{tabularx}
\end{table}

The direct-symbol route has a communication lower bound smaller by one factor $\ell$, but its upper bound on $L$ is also smaller by one factor $\ell$.  When $n$ is eliminated, the improvement in the rank estimate is squared, whereas the communication lower bound loses only one factor $\ell$.  This leaves a net improvement of one factor $\ell$: GPW obtains $\Omega(L^2/\ell^2)=\Omega((\log r)^2/(\log\log r)^2)$, while Theorem~\ref{thm:route_two_main} gives $\Omega(L^2/\ell)=\Omega((\log r)^2/\log\log r)$.  Thus the improvement comes entirely from avoiding Booleanization of the GPW input alphabet, and not from a larger multiplier in the communication simulation.

\section{The alphabet-preserving GPW lift}
\label{sec:lower_proof}

In this section, we combine the symbol-level GPW pointer function with a multicolor Index lifting theorem to obtain a stronger quantitative separation between deterministic communication complexity and logarithmic rank.  Specifically, we construct an infinite family of total Boolean communication matrices $(M_s)_{s\ge2}$ satisfying
$
  D(M_s)
  =
  \Omega(\frac{(\log \rank_{\R}(M_s))^2}
  {\log\log \rank_{\R}(M_s)}).
$
\subsection{Large-alphabet simulation}

Let $\Sigma$ be an alphabet of size $q$.  A query to a word in $\Sigma^P$ reveals one whole alphabet symbol.  Communication cost, in contrast, is measured in bits.  Define the $q$-ary Index gadget by
\begin{equation}
  \operatorname{Ind}_{K,\Sigma}\colon [K]\times\Sigma^K\to\Sigma,
  \qquad
  \operatorname{Ind}_{K,\Sigma}(a,y):=y_a.
  \label{eq:route_two_index}
\end{equation}

\begin{theorem}[Large-alphabet deterministic simulation]
\label{thm:route_two_lifting}
There is an absolute constant $P_0$ such that the following holds for every $P\ge P_0$. Let $\Sigma$ be an alphabet whose size $q:=|\Sigma|$ satisfies $2\le q\le P$, and set $K:=P^{1000}$. Then every $f\colon\Sigma^P\to\{0,1\}$ satisfies
\begin{equation*}
  \mathsf{P}^{\mathsf{cc}}(f\circ\operatorname{Ind}_{K,\Sigma}^P)
  \ge
  10\mathsf{P}^{\mathsf{dt}}(f)\log P.
\end{equation*}
\end{theorem}

\begin{proof}
Apply Theorem~\ref{thm:route_two_multicolor} with
$n:=P$ and $m:=K=P^{1000}$.
Then
$n=P=K^{1/1000}$ and $q\le P=K^{1/1000}$,
so all hypotheses of Theorem~\ref{thm:route_two_multicolor} hold once $K=P^{1000}\ge m_0$; set $P_0:=\lceil m_0^{1/1000}\rceil$.  
Theorem~\ref{thm:route_two_multicolor} gives
$\mathsf{P}^{\mathsf{cc}}(f\circ\operatorname{Ind}_{K,\Sigma}^P)\ge \frac{1}{100}\mathsf{P}^{\mathsf{dt}}(f)\log K=10\mathsf{P}^{\mathsf{dt}}(f)\log P$.
The last step uses $K=P^{1000}$, and hence $\log K=1000\log P$.
This proves the claim.
\end{proof}

\subsection{Unambiguous certificates for lifted functions}

The next ingredient counts the rectangles induced by symbol certificates.  An unambiguous family of $1$-certificates consists of pairwise disjoint certificate subcubes whose union is $f^{-1}(1)$. For a communication problem $F$, we write $\chi_1(F)$ for the minimum number of pairwise disjoint $1$-monochromatic rectangles whose union is $F^{-1}(1)$; in the notation of Section~\ref{sec:introduction}, $\chi_1(F)=p(M_F)$.

\begin{lemma}
\label{lem:route_two_up}
Let $f\colon\Sigma^P\to\{0,1\}$ admit an unambiguous family of $1$-certificates, each fixing at most $t$ coordinates.  If
$
  F:=f\circ\operatorname{Ind}_{K,\Sigma}^P,
$
then
\begin{enumerate}
\renewcommand{\theenumi}{(\alph{enumi})}
\renewcommand{\labelenumi}{\theenumi}
\item\label{part:route_two_partition}
$\chi_1(F)\le(qK)^t$.
\item\label{part:route_two_rank}
$\log\rank_{\R}(M_F)
\le
t(\log q+\log K)$.
\end{enumerate}
\end{lemma}

\begin{proof}
\noindent{\bf Proof of Part~\ref{part:route_two_partition}.}
Let $\mathcal C$ be the unambiguous certificate family.  Write a certificate as $C=(I_C,z_C)$, where $I_C\subseteq[P]$ and $z_C\in\Sigma^{I_C}$.  The corresponding subcube has uniform measure $q^{-|I_C|}$.  Since these subcubes are pairwise disjoint and $|I_C|\le t$,
$1\ge\sum_{C\in\mathcal C}q^{-|I_C|}\ge|\mathcal C|q^{-t}$.
Therefore
\begin{equation}
  |\mathcal C|\le q^t.
  \label{eq:route_two_certificate_count}
\end{equation}

For a lifted input
$x\in[K]^P$ and $y\in(\Sigma^K)^P$,
let its induced outer word $w(x,y)\in\Sigma^P$ be defined coordinatewise by
\begin{equation}
  w(x,y)_i:=y_i(x_i).
  \label{eq:route_two_induced_word}
\end{equation}
Fix $C=(I_C,z_C)\in\mathcal C$.  For each address vector $u\in[K]^{I_C}$, define
\[
  R_{C,u}
  :=
  \{x\in[K]^P:x_i=u_i\text{ for every }i\in I_C\}
  \times
  \{y\in(\Sigma^K)^P:y_i(u_i)=(z_C)_i\text{ for every }i\in I_C\}.
\]
This is a combinatorial rectangle.  If $(x,y)\in R_{C,u}$, then Eq.~\eqref{eq:route_two_induced_word} gives
$w(x,y)_i=(z_C)_i$ for every $i\in I_C$.
Thus $w(x,y)$ extends the accepting certificate $C$.  In particular, $F(x,y)=1$, and $R_{C,u}$ is $1$-monochromatic.

We next verify coverage and disjointness.  If $F(x,y)=1$, then $w(x,y)\in f^{-1}(1)$ and therefore extends a unique certificate $C\in\mathcal C$.  Taking $u=x|_{I_C}$ places $(x,y)$ in $R_{C,u}$, so the rectangles cover every $1$-input.  Conversely, suppose that $(x,y)$ belongs to both $R_{C,u}$ and $R_{C',u'}$.  Then $w(x,y)$ extends both $C$ and $C'$.  Unambiguity gives $C=C'$, and Alice's address vector then gives
$u=x|_{I_C}=u'$.
Thus the rectangles are pairwise disjoint.

A certificate $C$ contributes exactly $K^{|I_C|}$ rectangles. Therefore
\[
  \chi_1(F)
  \le
  \sum_{C\in\mathcal C}K^{|I_C|}
  \le
  |\mathcal C|K^t
  \le
  (qK)^t.
\]
The first step follows because the rectangles constructed above form a
disjoint $1$-monochromatic partition of $F^{-1}(1)$. The second step uses
$|I_C|\le t$ for every $C\in\mathcal C$, and the third step follows from
Eq.~\eqref{eq:route_two_certificate_count}.
This proves Lemma~\ref{lem:route_two_up}-\ref{part:route_two_partition}.

\medskip
\noindent{\bf Proof of Part~\ref{part:route_two_rank}.}
Because the rectangles form a disjoint partition of the $1$-entries, the communication matrix is the sum of their indicator matrices:
$M_F=\sum_{C\in\mathcal C}\sum_{u\in[K]^{I_C}}M_{R_{C,u}}$.
Every nonempty rectangle indicator is an outer product of two nonzero Boolean vectors and hence has real rank one. Applying the same rank-subadditivity argument to a minimum disjoint $1$-monochromatic rectangle partition gives the first step in
$\rank_{\R}(M_F)\le\chi_1(F)\le(qK)^t$,
and Lemma~\ref{lem:route_two_up}-\ref{part:route_two_partition} gives the second step.
Taking logarithms proves Lemma~\ref{lem:route_two_up}-\ref{part:route_two_rank}.
\end{proof}

\subsection{The GPW pointer function}

Following the notation of G\"o\"os, Pitassi, and Watson~\cite{gpw18}, we write $\mathsf{P}^{\mathsf{dt}}(f)$ for deterministic decision-tree complexity and $\mathsf{UP}^{\mathsf{dt}}(f)$ for the minimum width of an unambiguous family of $1$-certificates, the width of a family being the maximum number of coordinates fixed by any certificate in it.

Let $s\ge2$ be the side length of the GPW pointer grid and set
$N:=s^2$.
The symbol alphabet in~\cite[Section~2.1]{gpw18} is
\begin{equation}
  \Sigma_s
  :=
  \{0,1\}\times([s]\times[s]\cup\{\bot\}),
  \qquad
  q:=|\Sigma_s|=2(N+1).
  \label{eq:route_two_alphabet}
\end{equation}
Write the symbol in cell $(i,j)$ as
$(b_{i,j},p_{i,j})\in\Sigma_s$.
Order the cells of each column by their row index. For a column $j$ containing at least one nonnull pointer, define its pointer chain as follows: the chain starts at the cell targeted by the first nonnull pointer of column $j$, namely the pointer $p_{i,j}$ with $i$ minimal among the nonnull pointers of that column; from a cell with a nonnull pointer the chain moves to the targeted cell, and at a cell with a null pointer the chain ends. The cells reached by the chain are the cells visited by this process; this set is finite even when the pointers form a cycle, in which case the process revisits cells without visiting new ones. The pointer function $f_s\colon\Sigma_s^N\to\{0,1\}$ accepts precisely when there is a column $j$ with the following properties:
\begin{enumerate}
  \item Every bit $b_{i,j}$ in column $j$ equals one, and at least one pointer in that column is nonnull.
  \item Starting from the first nonnull pointer in column $j$, following the pointer chain reaches exactly one cell in every other column.
  \item Every cell reached by the pointer chain has bit value zero.
\end{enumerate}
The accepting column is unique because the pointer chain exhibits a zero in every other column.

Lemma \ref{lem:route_two_pointer_parameters} is a self-contained variant of the symbol-level pointer-function bounds from the “Actual gap example” in \cite[Section~2.1]{gpw18}. Our definition uses a slightly different convention for the terminal pointer.

\begin{lemma}[GPW pointer parameters, \cite{gpw18}]
\label{lem:route_two_pointer_parameters}
The function $f_s$ satisfies
\begin{enumerate}
\renewcommand{\theenumi}{(\alph{enumi})}
\renewcommand{\labelenumi}{\theenumi}
\item\label{part:route_two_depth}
$
  \mathsf{P}^{\mathsf{dt}}(f_s)=N,
$
\item\label{part:route_two_unambiguous}
$
  \mathsf{UP}^{\mathsf{dt}}(f_s)\le 2s-1.
$
\end{enumerate}
\end{lemma}
\begin{proof}
For the unambiguous upper bound, a certificate specifies the accepting column $j$.  It fixes the $s$ symbols in that column, thereby verifying that all their bit values are one and identifying the first nonnull pointer.  It then fixes the $s-1$ symbols reached by following the pointer chain through the remaining columns.  Thus the certificate fixes at most
$
  s+(s-1)=2s-1
$
symbols.  Every queried symbol is a complete pair in $\Sigma_s$, so the certificate fixes both fields even when verification uses only one field. Every extension of such a certificate is accepting: the chain computation of column $j$ reads only fixed symbols, since the first nonnull pointer of column $j$ is determined by the $s$ fixed column symbols and every subsequent step reads the pointer of a previously fixed chain cell. Hence the cells reached on any extension are exactly the $s-1$ fixed chain cells, properties 1--3 depend only on fixed symbols, and every extension accepts via column $j$. On an accepting input, the column $j$ is unique, the first nonnull pointer in that column is unique, and every subsequent cell is determined by a previously fixed pointer.  Hence exactly one certificate accepts the input.  This proves Lemma~\ref{lem:route_two_pointer_parameters}-\ref{part:route_two_unambiguous}.

For the deterministic lower bound, call a query critical if it asks for the last unqueried cell in its column.  The adversary answers every noncritical query by $(1,\bot)$.  It answers the first critical query by $(0,\bot)$.  At every later critical query, it answers $(0,p)$, where $p$ points to the cell of the preceding critical query.

Suppose that only one of the $N$ cells remains unqueried.  Every completed column contains exactly one critical zero-cell, while all previously queried cells in the remaining column equal $(1,\bot)$.  There are two completions consistent with the entire transcript.  Setting the last cell to $(0,\bot)$ leaves a zero in every column, so the function value is zero.  Alternatively, set the last cell to $(1,p)$, where $p$ points to the most recent critical cell.  The remaining column is then the unique all-one column, and its only nonnull pointer starts the reverse chronological chain of critical cells.  This chain visits one zero-cell in every other column and therefore makes the function value one.  Thus the value remains undetermined after $N-1$ queries, so
$
  \mathsf{P}^{\mathsf{dt}}(f_s)\ge N.
$
Every shorter transcript produced by the adversary can be continued to such an $(N-1)$-query transcript, so it also has both a zero-completion and a one-completion.  Querying all $N$ cells gives the matching upper bound.  This proves Lemma~\ref{lem:route_two_pointer_parameters}-\ref{part:route_two_depth}.
\end{proof}

\subsection{Padding and lifting}

Set
$P:=4N$ and $K:=P^{1000}$.
Fix any symbol $a_s\in\Sigma_s$, and extend $f_s$ to
\[
  \widetilde f_s\colon\Sigma_s^P\to\{0,1\},
  \qquad
  \widetilde f_s(z_1,\ldots,z_P):=f_s(z_1,\ldots,z_N).
\]
The upper bound
$\mathsf{P}^{\mathsf{dt}}(\widetilde f_s)\le\mathsf{P}^{\mathsf{dt}}(f_s)$
follows by ignoring the dummy coordinates.  Conversely, restricting every dummy coordinate to $a_s$ converts any decision tree for $\widetilde f_s$ into one for $f_s$.  Hence
\begin{equation}
  \mathsf{P}^{\mathsf{dt}}(\widetilde f_s)=N.
  \label{eq:route_two_padded_depth}
\end{equation}
The certificate family proving Lemma~\ref{lem:route_two_pointer_parameters}-\ref{part:route_two_unambiguous} leaves all dummy coordinates free, and therefore
\begin{equation}
  \mathsf{UP}^{\mathsf{dt}}(\widetilde f_s)\le2s-1.
  \label{eq:route_two_padded_up}
\end{equation}
Moreover, Eq.~\eqref{eq:route_two_alphabet} gives
$q=2(N+1)\le4N=P$.
Here the inequality uses $N=s^2\ge1$, and the last equality follows from $P:=4N$.
Thus all hypotheses of Theorem~\ref{thm:route_two_lifting} hold.

\begin{definition}[Lifted GPW communication problem]
\label{def:route_two_function}
Define the total Boolean communication problem
\[
  F_s
  :=
  \widetilde f_s\circ\operatorname{Ind}_{K,\Sigma_s}^P.
\]
\end{definition}

The family is explicit: for any $s$ and any pair $x\in[K]^P$, $y\in(\Sigma_s^K)^P$, the entry $M_s[x,y]=F_s(x,y)$ is computable in time polynomial in the bit length of $(x,y)$, by evaluating the induced word of Eq.~\eqref{eq:route_two_induced_word} coordinatewise and then evaluating $f_s$ along its pointer chain.

\begin{theorem}[Lifted GPW separation]
\label{thm:route_two_main}
Let $M_s:=M_{F_s}$ be the communication matrix of the problem in Definition~\ref{def:route_two_function}.  Then the following statements hold.
\begin{enumerate}
\renewcommand{\theenumi}{(\alph{enumi})}
\renewcommand{\labelenumi}{\theenumi}
\item\label{part:route_two_comm_lower} $\mathsf{P}^{\mathsf{cc}}(F_s)=\Omega(s^2\log s)$.
\item\label{part:route_two_rank_upper} $\log\rank_{\R}(M_s)=O(s\log s)$.
\item\label{part:route_two_final} For all sufficiently large $s$,
$
  D(M_s)
  =
  \Omega(\frac{(\log \rank_{\R}(M_s))^2}
  {\log\log \rank_{\R}(M_s)}).
$
\end{enumerate}
\end{theorem}

\begin{proof}[{\bf Proof of Part~\ref{part:route_two_comm_lower}}]
Apply Theorem~\ref{thm:route_two_lifting} to $\widetilde f_s$.  Since $P=4N$, $N=s^2$, and Eq.~\eqref{eq:route_two_padded_depth} holds, we obtain
$\mathsf{P}^{\mathsf{cc}}(F_s)=\Omega(N\log P)=\Omega(s^2\log s)$.
The last step uses $P=4s^2$, and hence $\log P=\Theta(\log s)$.
This proves Theorem~\ref{thm:route_two_main}-\ref{part:route_two_comm_lower}.
\end{proof}

\begin{proof}[{\bf Proof of Part~\ref{part:route_two_rank_upper}}]
Apply Lemma~\ref{lem:route_two_up}-\ref{part:route_two_rank} with $t:=2s-1$, using Eq.~\eqref{eq:route_two_padded_up}.  Since
$q=2(s^2+1)$ and $K=(4s^2)^{1000}$,
Lemma~\ref{lem:route_two_up}-\ref{part:route_two_rank} gives
$\log\rank_{\R}(M_s)\le(2s-1)(\log q+\log K)=O(s\log s)$.
The last step uses $\log q=\Theta(\log s)$ and $\log K=1000\log(4s^2)=\Theta(\log s)$.
This proves Theorem~\ref{thm:route_two_main}-\ref{part:route_two_rank_upper}.
\end{proof}

\begin{proof}[{\bf Proof of Part~\ref{part:route_two_final}}]
It remains to eliminate $s$.  Write
$r_s:=\rank_{\R}(M_s)$ and $L_s:=\log r_s$.
We first note that $L_s\to\infty$.  Indeed, a Boolean matrix of real rank $r$ has at most $2^r$ distinct rows.  To see this, choose $r$ columns forming a basis of its column space.  Two rows that agree on those basis columns agree on every column and are therefore identical.  Alice can send the identity of her row using at most $r$ bits, and Bob then announces the function value using one further bit, so
$D(M_s)=\mathsf{P}^{\mathsf{cc}}(F_s)\le r_s+1$.
Theorem~\ref{thm:route_two_main}-\ref{part:route_two_comm_lower} now implies $r_s\to\infty$ and hence $L_s\to\infty$.

Choose a constant $C>0$ such that
$L_s\le Cs\log s$
for all sufficiently large $s$, as permitted by Theorem~\ref{thm:route_two_main}-\ref{part:route_two_rank_upper}.  The function
$\phi(x):=x^2/\log x$
is increasing for all sufficiently large $x$.  Since $L_s\to\infty$, for all sufficiently large $s$ we therefore have
\[
  \frac{L_s^2}{\log L_s}
  =
  \phi(L_s)
  \le
  \phi(Cs\log s)
  =
  O(s^2\log s).
\]
The last step uses $\log(Cs\log s)=\Theta(\log s)$.
Combining this estimate with Theorem~\ref{thm:route_two_main}-\ref{part:route_two_comm_lower} and the identity $D(M_s)=\mathsf{P}^{\mathsf{cc}}(F_s)$ gives
\[
  D(M_s)
  =
  \Omega(\frac{L_s^2}{\log L_s})
  =
  \Omega(\frac{(\log\rank_{\R}(M_s))^2}
  {\log\log\rank_{\R}(M_s)}),
\]
The last step substitutes $L_s=\log r_s$ and $r_s=\rank_{\R}(M_s)$.
This proves Theorem~\ref{thm:route_two_main}-\ref{part:route_two_final}.
\end{proof}

\begin{corollary}[Partition number separation]
\label{cor:route_two_partition}
Let $p(M_s)$ denote the minimum number of pairwise disjoint $1$-monochromatic rectangles whose union is the set of $1$-entries of $M_s$. Then, for all sufficiently large $s$,
\[
 D(M_s)
 =
 \Omega(\frac{(\log p(M_s))^2}{\log\log p(M_s)}).
\]
\end{corollary}

\begin{proof}
Apply Lemma~\ref{lem:route_two_up}-\ref{part:route_two_partition} with $t:=2s-1$, using Eq.~\eqref{eq:route_two_padded_up}, to obtain
\[
 p(M_s)=\chi_1(F_s)\le (qK)^{2s-1}.
\]
Taking logarithms and substituting $q=2(s^2+1)$ and $K=(4s^2)^{1000}$ as in the proof of Theorem~\ref{thm:route_two_main}-\ref{part:route_two_rank_upper} gives $\log p(M_s)=O(s\log s)$. Moreover $\rank_{\R}(M_s)\le p(M_s)$ by Lemma~\ref{lem:route_two_up}-\ref{part:route_two_rank}, so $\log p(M_s)\to\infty$. The elimination of $s$ now proceeds exactly as in the proof of Theorem~\ref{thm:route_two_main}-\ref{part:route_two_final}, with $L_s:=\log p(M_s)$ in place of $\log\rank_{\R}(M_s)$.
\end{proof}

\begin{corollary}[Clique versus independent set]
\label{cor:route_two_cis}
There is an explicit infinite family of graphs $(G_s)_{s\ge2}$, where $G_s$ has $n_s$ vertices, such that
\[
 D(\operatorname{CIS}_{G_s})
 =
 \Omega(\frac{(\log n_s)^2}{\log\log n_s}).
\]
\end{corollary}

\begin{proof}
Let $\mathcal R$ be the family of pairwise disjoint $1$-monochromatic rectangles constructed in the proof of Lemma~\ref{lem:route_two_up}-\ref{part:route_two_partition}, and let $G_s$ be the graph whose vertices are the members of $\mathcal R$, two of them being adjacent when they share a row. Alice maps her input $x$ to the set of rectangles meeting row $x$; any two of these share the row $x$, so this set is a clique. Bob maps his input $y$ to the set of rectangles meeting column $y$; two such rectangles cannot share a row, since a common row $x$ would place $(x,y)$ in both, contradicting disjointness, so this set is an independent set. The intersection of the two sets consists exactly of the rectangles containing $(x,y)$, of which there is one when $F_s(x,y)=1$ and none otherwise. Both maps are computed without communication, so $F_s$ is a subproblem of $\operatorname{CIS}_{G_s}$ and
\[
 D(\operatorname{CIS}_{G_s})\ge\mathsf{P}^{\mathsf{cc}}(F_s).
\]
The graph has $n_s=|\mathcal R|\le (qK)^{2s-1}$ vertices, so $\log n_s=O(s\log s)$ as in Corollary~\ref{cor:route_two_partition}, while Theorem~\ref{thm:route_two_main}-\ref{part:route_two_comm_lower} gives $\mathsf{P}^{\mathsf{cc}}(F_s)=\Omega(s^2\log s)$. Eliminating $s$ as in Theorem~\ref{thm:route_two_main}-\ref{part:route_two_final} proves the claim.
\end{proof}

\paragraph{Comparison with the GPW parameters.}
The error-correcting-code Booleanization used by GPW, followed by Boolean Index lifting, gives the denominator $(\log\log r)^2$ in Eq.~\eqref{eq:current_lower}; see~\cite[Section~4]{gpw18}.  The construction above instead lifts the pointer symbols directly.  Its deterministic lower bound is $\Omega(s^2\log s)$ rather than $\Omega(s^2\log^2 s)$, while its logarithmic rank upper bound is $O(s\log s)$ rather than $O(s\log^2 s)$.  Removing the Booleanization overhead therefore eliminates one factor of $\log\log r$ and yields Theorem~\ref{thm:route_two_main}-\ref{part:route_two_final}.  Within this construction, a quadratic lower bound would require eliminating the remaining Index-gadget overhead or proving a sharper rank bound for the lifted matrix.

\section{The multicolor simulation theorem}
\label{sec:route_two_multicolor}

We prove the large-alphabet deterministic simulation theorem used in
Theorem~\ref{thm:route_two_lifting}. Roughgarden and Weinstein state a
simulation theorem of this type for search problems over an arbitrary
alphabet~\cite[Theorem~5.5]{rw16}, attributing it to Raz and
McKenzie~\cite{rm99} and to G\"o\"os, Pitassi, and
Watson~\cite{gpw18}. With the gadget size chosen as in the paragraph
following their theorem, the statement directly covers the lifting
application used here. However, the cited proofs do not present this
result in the required form: G\"o\"os, Pitassi, and Watson give a formal
proof for the Boolean input alphabet, while Raz and McKenzie treat the
multicolor setting for structured search problems. We therefore include
a self-contained proof for the parameter regime needed here.

All logarithms are base two.  Let $\Sigma$ be an alphabet of size $q$.
Define
$\operatorname{Ind}_{m,\Sigma}(a,y):=y_a$ for $a\in[m]$ and $y\in\Sigma^m$,
and extend this gadget coordinatewise.

\begin{theorem}[Multicolor deterministic simulation]
\label{thm:route_two_multicolor}
There is an absolute constant $m_0$ such that the following holds.  Let
$\Sigma$ be an alphabet of size $q$, and suppose
\[
  m\ge m_0,
  \qquad
  1\le n\le m^{1/1000},
  \qquad
  2\le q\le m^{1/1000}.
\]
For every total function $f\colon\Sigma^n\to\{0,1\}$,
\begin{equation*}
  \mathsf{P}^{\mathsf{cc}}(f\circ\operatorname{Ind}_{m,\Sigma}^n)
  \ge
  \frac{1}{100}\mathsf{P}^{\mathsf{dt}}(f)\log m.
\end{equation*}
\end{theorem}

\subsection{Deficiency, projections, and thickness}
For nonempty sets $A\subseteq[m]^s$ and $B\subseteq(\Sigma^m)^s$, define
$\alpha(A):=\log(m^s/|A|)$ and $\beta(B):=\log(q^{ms}/|B|)$.
For $J\subseteq[s]$, let $A_J$ and $B_J$ denote the corresponding
coordinate projections.  In particular, $A_{-i}:=A_{[s]\setminus\{i\}}$.
We regard $[m]^0$ and $(\Sigma^m)^0$ as singleton sets, and thickness in
dimension zero is vacuous.  View every $a\in A$ as an edge between
$a_i\in[m]$ and $a_{-i}\in[m]^{s-1}$.  Define
$\operatorname{AvgDeg}_i(A):=|A|/|A_{-i}|$
and let $\operatorname{MinDeg}_i(A)$ be the minimum nonzero fiber size
over $A_{-i}$.  We call $A$ \emph{thick} if
$\operatorname{MinDeg}_i(A)\ge m^{17/20}$ for every $i\in[s]$.
For $U\subseteq[m]$ and $V\subseteq\Sigma^m$, write
$A_{i,U}:=\{a\in A:a_i\in U\}$ and $B_{i,V}:=\{b\in B:b_i\in V\}$.

\begin{lemma}[Thickness cleanup]
\label{lem:route_two_thickness}
Suppose $A\subseteq[m]^s$ is nonempty and
$\operatorname{AvgDeg}_i(A)\ge d$ for every $i\in[s]$. There is a
set $A'\subseteq A$ such that $|A'|\ge |A|/2$ and
$\operatorname{MinDeg}_i(A')\ge d/(2s)$ for every $i\in[s]$.
Consequently, $\alpha(A')\le\alpha(A)+1$.
\end{lemma}

\begin{proof}
Repeatedly choose a coordinate $i$ and a nonempty fiber of size less
than $d/(2s)$, and delete that fiber. Once a fiber over some element of
$A_{-i}$ is deleted it stays empty, since points are only ever
removed, so the direction-$i$ deletions are indexed by distinct
elements of the original projection $A_{-i}$.  For a fixed $i$, the number of
fibers ever deleted is at most
$|A_{-i}|=|A|/\operatorname{AvgDeg}_i(A)\le|A|/d$.
Thus fewer than
$\sum_{i=1}^s(|A|/d)(d/(2s))=|A|/2$
points are deleted.  The process therefore terminates with the stated
set $A'$.
\end{proof}

\subsection{A multicolor fiber-hitting lemma}

We next prove the multicolor estimate that replaces the binary
entropy calculation in the GPW projection lemma.  Fix
$\sigma\in\Sigma$ and set
$u:=\lfloor m^{7/20}\rfloor$.
For $U\subseteq[m]$, define
$V_\sigma(U):=\{y\in\Sigma^m:y_j=\sigma\text{ for every }j\in U\}$.

\begin{lemma}[Multicolor fiber hitting]
\label{lem:route_two_fiber_hitting}
Suppose $2\le q\le m^{1/1000}$ and $m$ is sufficiently large.
If a nonempty set $W\subseteq\Sigma^m$ satisfies
\[
  \beta(W):=\log\frac{q^m}{|W|}
  \le m^{11/20},
\]
then a uniformly random $u$-element set $U\subseteq[m]$ satisfies
\begin{equation*}
  \Pr_U[V_\sigma(U)\cap W\ne\varnothing]\ge\frac34.
\end{equation*}
\end{lemma}

\begin{proof}
For a nonempty $T\subseteq\Sigma^m$, define
$\operatorname{Bad}_\sigma(T):=\{j\in[m]:|\{y\in T:y_j=\sigma\}|<|T|/(2q)\}$.
We first claim that
\begin{equation}
  |\operatorname{Bad}_\sigma(T)|
  \le
  3q^2\beta(T).
  \label{eq:route_two_bad_coordinates}
\end{equation}
Let $Y$ be uniform on $T$, let $p_j$ be the distribution of $Y_j$,
and let $u_\Sigma$ be the uniform distribution on $\Sigma$.  If $j$ is
bad, then
$u_\Sigma(\sigma)-p_j(\sigma)>1/(2q)$.
Because the coordinate differences between $p_j$ and $u_\Sigma$ sum to
zero,
$\lVert p_j-u_\Sigma\rVert_1\ge1/q$.
Pinsker's inequality in base-two units, together with the identity
$D_{\mathrm{KL}}^{(2)}(p_j\|u_\Sigma)=\log q-H(Y_j)$, gives
\[
 \log q-H(Y_j)
 =D_{\mathrm{KL}}^{(2)}(p_j\|u_\Sigma)
 \ge\frac{\log e}{2}\lVert p_j-u_\Sigma\rVert_1^2
 \ge\frac{\log e}{2q^2}
 \ge\frac{1}{3q^2}.
\]
Subadditivity of entropy now gives
\[
  m\log q-\beta(T)
  =
  H(Y)
  \le
  \sum_{j=1}^m H(Y_j)
  \le
  m\log q-
  \frac{|\operatorname{Bad}_\sigma(T)|}{3q^2},
\]
which proves Eq.~\eqref{eq:route_two_bad_coordinates}.

Sample $U:=\{j_1,\ldots,j_u\}$ sequentially without replacement.  Set
$W_0:=W$ and $W_t:=\{y\in W:y_{j_1}=\cdots=y_{j_t}=\sigma\}$.
Condition on the event that
$j_r\notin\operatorname{Bad}_\sigma(W_{r-1})$ for every $r\le t$.
Then
\[
  \beta(W_t)
  \le
  \beta(W)+t\log(2q)
  \le
  m^{11/20}+u\log(2q).
\]
Using $q\le m^{1/1000}$ and Eq.~\eqref{eq:route_two_bad_coordinates},
for all sufficiently large $m$ we have
$|\operatorname{Bad}_\sigma(W_t)|\le m^{14/25}$.
Since $u=o(m)$, the conditional probability that the next sampled
coordinate is bad is at most $2m^{-11/25}$.  A union bound over the
$u\le m^{7/20}$ choices gives total failure probability at most
$2m^{7/20-11/25}=2m^{-9/100}<1/4$.
If no bad coordinate is selected, every $W_t$ is nonempty, and hence
$V_\sigma(U)\cap W=W_u\ne\varnothing$.  This proves Lemma~\ref{lem:route_two_fiber_hitting}.
\end{proof}

\subsection{Projection and realization of the final symbol}

\begin{lemma}[Multicolor projection lemma]
\label{lem:route_two_projection}
Suppose that $m$ is sufficiently large,
\[
  1\le s\le m^{1/1000},
  \qquad
  2\le q\le m^{1/1000},
\]
$A\subseteq[m]^s$ is nonempty and thick, and the nonempty set
$B\subseteq(\Sigma^m)^s$ satisfies $\beta(B)\le m^{1/10}$.
For every $i\in[s]$ and $\sigma\in\Sigma$, there is an
$\sigma$-monochromatic rectangle
\[
  U\times V\subseteq[m]\times\Sigma^m
\]
for $\operatorname{Ind}_{m,\Sigma}$ such that
\begin{enumerate}
\renewcommand{\theenumi}{(\alph{enumi})}
\renewcommand{\labelenumi}{\theenumi}
\item\label{part:route_two_projection_thick}
$(A_{i,U})_{-i}$ is thick.
\item\label{part:route_two_projection_alpha}
$\alpha((A_{i,U})_{-i})
\le
\alpha(A)-\log m+\log\operatorname{AvgDeg}_i(A)$.
\item\label{part:route_two_projection_beta}
$\beta((B_{i,V})_{-i})
\le
\beta(B)+1$.
\end{enumerate}
Under the convention above, thickness in
Lemma~\ref{lem:route_two_projection}-\ref{part:route_two_projection_thick}
is vacuous when $s=1$.
\end{lemma}

\begin{proof}
By relabeling coordinates, assume $i=s$.  Choose a uniformly random
$u$-element set $U\subseteq[m]$ and set $V:=V_\sigma(U)$.  Then
$U\times V$ is $\sigma$-monochromatic.

We first show that, except with probability less than
$2^{-m^{3/20}}$,
\begin{equation}
  (A_{s,U})_{-s}=A_{-s}.
  \label{eq:route_two_full_projection}
\end{equation}
Every nonempty fiber of $A$ over $A_{-s}$ has size at least
$m^{17/20}$.  The probability that $U$ misses a fixed such fiber is
at most
$\exp(-m^{4/20}/2)$.
There are at most $m^{s-1}$ fibers.  Since
$s\le m^{1/1000}$, a union bound proves
Eq.~\eqref{eq:route_two_full_projection} with the asserted failure
probability for all sufficiently large $m$.

\medskip
\noindent{\bf Proof of Part~\ref{part:route_two_projection_thick}.}
If Eq.~\eqref{eq:route_two_full_projection} holds, projection cannot
decrease any of the remaining minimum degrees, so
Lemma~\ref{lem:route_two_projection}-\ref{part:route_two_projection_thick} holds.

\medskip
\noindent{\bf Proof of Part~\ref{part:route_two_projection_alpha}.}
Moreover,
$|A_{-s}|=|A|/\operatorname{AvgDeg}_s(A)$,
and hence
$\alpha(A_{-s})=\alpha(A)-\log m+\log\operatorname{AvgDeg}_s(A)$.
This proves Lemma~\ref{lem:route_two_projection}-\ref{part:route_two_projection_alpha}.

\medskip
\noindent{\bf Proof of Part~\ref{part:route_two_projection_beta}.}
It remains to find some $U$ that also satisfies
Eq.~\eqref{eq:route_two_full_projection}.  Set
$d:=|B|/q^{ms}=2^{-\beta(B)}$.
For each $b_{-s}\in(\Sigma^m)^{s-1}$, let
$W_{b_{-s}}:=\{b_s\in\Sigma^m:(b_{-s},b_s)\in B\}$.
Lemma~\ref{lem:route_two_fiber_hitting} implies, for every
$W\subseteq\Sigma^m$,
\begin{equation}
  \Pr_U[V_\sigma(U)\cap W\ne\varnothing]
  \ge
  \frac34\frac{|W|}{q^m}-2^{-m^{11/20}}.
  \label{eq:route_two_fiber_all_sizes}
\end{equation}
The claim is immediate when $W=\varnothing$.  Suppose that $W$ is
nonempty.  If $\beta(W)\le m^{11/20}$, the left-hand side is at least
$3/4$ by Lemma~\ref{lem:route_two_fiber_hitting}.  Otherwise,
$|W|/q^m<2^{-m^{11/20}}$, so the right-hand side is negative.

Let
$X(U):=|(B_{s,V_\sigma(U)})_{-s}|/q^{m(s-1)}$.
Averaging Eq.~\eqref{eq:route_two_fiber_all_sizes} over all fibers
gives
\[
  \E_U[X(U)]
  \ge
  \frac34d-2^{-m^{11/20}}
  \ge
  \frac58d,
\]
where the last inequality uses
$d\ge2^{-m^{1/10}}$.  Since $0\le X(U)\le1$,
$\Pr_U[X(U)\ge d/2]\ge d/8$.
For sufficiently large $m$,
$d/8\ge2^{-m^{1/10}-3}>2^{-m^{3/20}}$.
Let $\mathcal E_{\mathrm{full}}$ denote the event in
Eq.~\eqref{eq:route_two_full_projection}, and let
$\mathcal E_{\mathrm{hit}}:=\{X(U)\ge d/2\}$. Then inclusion--exclusion gives
$\Pr_U[\mathcal E_{\mathrm{full}}\cap\mathcal E_{\mathrm{hit}}]
\ge\Pr_U[\mathcal E_{\mathrm{hit}}]-\Pr_U[\mathcal E_{\mathrm{full}}^c]$.
The preceding bounds show that the right-hand side is at least
$d/8-2^{-m^{3/20}}>0$. Thus some $U$ satisfies both events.  For this choice,
\[
  \beta((B_{s,V})_{-s})
  =
  -\log X(U)
  \le
  -\log(d/2)
  =
  \beta(B)+1,
\]
which proves Lemma~\ref{lem:route_two_projection}-\ref{part:route_two_projection_beta}.
\end{proof}

\begin{lemma}[Realizing the final symbol]
\label{lem:route_two_last_color}
Suppose $S\subseteq[m]$ and a nonempty set $T\subseteq\Sigma^m$ satisfy
\[
  |S|\ge m^{17/20},
  \qquad
  \beta(T)\le m^{1/10},
  \qquad
  2\le q\le m^{1/1000}.
\]
For every $\sigma\in\Sigma$, there are $a\in S$ and $y\in T$ such
that $y_a=\sigma$.
\end{lemma}

\begin{proof}
Otherwise every $y\in T$ avoids $\sigma$ on all coordinates in $S$,
so
$|T|\le q^m(1-1/q)^{|S|}$.
It follows that
\[
  \beta(T)
  \ge
  |S|\log\frac{q}{q-1}
  \ge
  \frac{|S|\log e}{q}
  \ge
  m^{849/1000}\log e,
\]
contradicting $\beta(T)\le m^{1/10}$ for sufficiently large $m$.
\end{proof}

\subsection{Proof of the multicolor simulation theorem}
\begin{proof}[Proof of Theorem~\ref{thm:route_two_multicolor}]
Let
$F:=f\circ\operatorname{Ind}_{m,\Sigma}^n$
and let $c:=\mathsf{P}^{\mathsf{cc}}(F)$.  If
$c\ge n\log m/100$,
then the conclusion of Theorem~\ref{thm:route_two_multicolor} follows from
$\mathsf{P}^{\mathsf{dt}}(f)\le n$.  We may therefore assume
\begin{equation}
  c<\frac{n\log m}{100}.
  \label{eq:route_two_low_communication}
\end{equation}
Fix a deterministic protocol of cost $c$ for $F$.

We construct a decision tree for $f$.  On an unknown input
$z\in\Sigma^n$, maintain a protocol node $v$, a set $I\subseteq[n]$
of unqueried coordinates, and nonempty sets
$A\subseteq[m]^n$ and $B\subseteq(\Sigma^m)^n$.
Initially, $v$ is the protocol root, $I:=[n]$, and $A,B$ are the full
spaces.  We maintain the following invariants:
\begin{enumerate}
  \item $A_I$ is thick.
  \item $A\times B$ is contained in the protocol rectangle at $v$.
  \item For every $j\notin I$ and every $(a,b)\in A\times B$,
$\operatorname{Ind}_{m,\Sigma}(a_j,b_j)=z_j$.
\end{enumerate}

While $v$ is not a protocol leaf and $I$ is nonempty, perform one of
the following two steps.

If
$\operatorname{AvgDeg}_i(A_I)\ge m^{19/20}$ for every $i\in I$,
simulate the next protocol bit.  If Alice speaks, then for
$\varepsilon\in\{0,1\}$, let $A^\varepsilon\subseteq A$ be the set of
Alice inputs that send $\varepsilon$ at $v$. Choose $\varepsilon$ so
that, for $C:=(A^\varepsilon)_I$, we have $|C|\ge |A_I|/2$.
For every $i\in I$,
\[
  \operatorname{AvgDeg}_i(C)
  =
  \frac{|C|}{|C_{-i}|}
  \ge
  \frac{|A_I|/2}{|(A_I)_{-i}|}
  =
  \frac12\operatorname{AvgDeg}_i(A_I)
  \ge
  \frac12m^{19/20}.
\]
Apply Lemma~\ref{lem:route_two_thickness} to $C$ with
$d:=m^{19/20}/2$, obtaining $C'\subseteq C$.  Since $|I|\le m^{1/1000}$,
the quantity $m^{19/20}/(4|I|)$ is at least $m^{17/20}$ for sufficiently
large $m$, so $C'$ is thick.  Restrict $A$ to the $\varepsilon$-branch so that $A_I=C'$,
and move $v$ to the $\varepsilon$-child.  This communication step increases
$\alpha(A_I)$ by at most two.

If Bob speaks, let $B^\varepsilon\subseteq B$ be the set of Bob
inputs that send $\varepsilon$ at $v$, and choose $\varepsilon$ such that
$|(B^\varepsilon)_I|\ge|B_I|/2$.
Replace $B$ by $B^\varepsilon$ and move $v$ to the
$\varepsilon$-child.  This communication step increases $\beta(B_I)$
by at most one.  In either case, the three invariants are preserved.

Otherwise, choose $i\in I$ with
$\operatorname{AvgDeg}_i(A_I)<m^{19/20}$
and query $z_i$. Apply
Lemma~\ref{lem:route_two_projection} to $A_I,B_I$ with
$\sigma:=z_i$, obtaining $U\subseteq[m]$ and $V\subseteq\Sigma^m$. Define
$A':=\{a\in A:a_i\in U\}$, $B':=\{b\in B:b_i\in V\}$, and
$I':=I\setminus\{i\}$. Lemma~\ref{lem:route_two_projection}-\ref{part:route_two_projection_thick}
shows that $A'_{I'}$ is thick. Containment in the current protocol
rectangle is preserved because $A'\times B'\subseteq A\times B$.
Moreover, $U\times V$ is $z_i$-monochromatic, so the third invariant
holds for the newly queried coordinate $i$ and remains true for the
previously queried coordinates by restriction.
Lemma~\ref{lem:route_two_projection}-\ref{part:route_two_projection_alpha}
and the choice of $i$ give
\begin{equation}
  \alpha(A'_{I'})
  \le
  \alpha(A_I)-\frac{\log m}{20},
  \label{eq:route_two_potential_drop}
\end{equation}
while Lemma~\ref{lem:route_two_projection}-\ref{part:route_two_projection_beta}
gives $\beta(B'_{I'})\le\beta(B_I)+1$. Finally, replace
$(A,B,I)$ by $(A',B',I')$.

Fix an arbitrary root-to-leaf path of the constructed decision tree.
For a prefix of this path in which every query step performed so far is
valid, let $Q_t$ denote the number of such query steps. There are at most
$c$ communication steps in the prefix. The potential $\alpha(A_I)$ starts
at zero, remains nonnegative, increases by at most two in each Alice
communication step, and decreases by at least $\log m/20$ in every query
step. Therefore,
$Q_t\le 40c/\log m$.

We prove by induction that every query step on the path is valid. Before a
prospective query, suppose that all preceding query steps are valid. The
prefix bound above applies. Since each Bob communication step and each
preceding query step increases $\beta(B_I)$ by at most one,
we have
$\beta(B_I)\le c+Q_t\le c+40c/\log m$.
By Eq.~\eqref{eq:route_two_low_communication} and $n\le m^{1/1000}$,
for all sufficiently large $m$ this quantity is less than
$n\log m/100+2n/5<m^{1/10}$. Hence the hypothesis
$\beta(B_I)\le m^{1/10}$ of Lemma~\ref{lem:route_two_projection}
holds, so the next query step is valid. This completes the induction.

Let $Q$ denote the total number of query steps on the chosen root-to-leaf
path. Applying the prefix bound to the complete path gives
\begin{equation}
  Q\le\frac{40c}{\log m}.
  \label{eq:route_two_query_count}
\end{equation}
Since the path was arbitrary, Eq.~\eqref{eq:route_two_query_count} holds
on every root-to-leaf path. Moreover, throughout each path,
$\beta(B_I)\le c+Q<m^{1/10}$.

If all coordinates are queried, the decision tree outputs $f(z)$.
Suppose instead that the simulation reaches a protocol leaf while
$I\ne\varnothing$.  The decision tree outputs the label of this leaf.
To prove that this output is correct, fix an arbitrary word
$w\in\Sigma^n$ consistent with the query answers along the current
path. While $|I|>1$, choose a coordinate $i\in I$ and apply
Lemma~\ref{lem:route_two_projection} to $A_I,B_I$ with
$\sigma:=w_i$, obtaining $U\subseteq[m]$ and $V\subseteq\Sigma^m$.
Replace
$A:=\{a\in A:a_i\in U\}$, $B:=\{b\in B:b_i\in V\}$, and
$I:=I\setminus\{i\}$. Lemma~\ref{lem:route_two_projection}-\ref{part:route_two_projection_thick}
preserves thickness, while
Lemma~\ref{lem:route_two_projection}-\ref{part:route_two_projection_beta}
increases the deficiency of the projected Bob set by at most one per
restriction. These auxiliary restrictions are not operations
performed by the decision tree. There are at most $n-1$ such
restrictions.  By Eq.~\eqref{eq:route_two_low_communication} and
Eq.~\eqref{eq:route_two_query_count}, the deficiency throughout this
continuation is at most
\[
  c+Q+n
  \le
  \frac{n\log m}{100}+\frac{2n}{5}+n
  <
  m^{1/10}
\]
for all sufficiently large $m$.

At the end of the auxiliary restriction process, write $I=\{i\}$.
Thickness gives $|A_{\{i\}}|\ge m^{17/20}$, while
$\beta(B_{\{i\}})\le m^{1/10}$.  By
Lemma~\ref{lem:route_two_last_color}, there are
$a_i\in A_{\{i\}}$ and $b_i\in B_{\{i\}}$ satisfying
$\operatorname{Ind}_{m,\Sigma}(a_i,b_i)=w_i$.
Choose $a\in A$ extending the selected $a_i$, and choose $b\in B$
extending the selected $b_i$.  Such extensions exist by the definitions
of $A_{\{i\}}$ and $B_{\{i\}}$.  Since the maintained set is the
Cartesian product $A\times B$, the pair $(a,b)$ belongs to it.  The
invariants and the auxiliary restrictions imply
$\operatorname{Ind}_{m,\Sigma}^n(a,b)=w$.
Moreover, $(a,b)$ lies in the protocol rectangle of the reached leaf.
The leaf label is therefore $F(a,b)=f(w)$.  Since $w$ was arbitrary,
every completion consistent with the decision-tree path has the same
$f$-value, and the constructed decision tree is correct.

Eq.~\eqref{eq:route_two_query_count} now gives
$\mathsf{P}^{\mathsf{dt}}(f)\le40c/\log m$.
Together with the high-communication case, this completes the proof of
Theorem~\ref{thm:route_two_multicolor}.
\end{proof}

\bibliographystyle{alpha}
\bibliography{ref}

\appendix

%%% This file is the structure for appendix content
%%% TeX files for body contents should be named as:
%%% 50_xxxx.tex
%%% 51_xxxx.tex
%%% ...

%\input{50_roadmap}

\section*{Acknowledgments}

The author used Codex 5.6 Sol and Claude Code Fable 5 to assist with language editing and grammar checking.
\section{Implication from BBGJK}
\label{sec:bbgjk}

Balodis, Ben-David, G\"o\"os, Jain, and Kothari~\cite{bbgjk23} constructed graphs whose chromatic number is exponential in a near-quadratic power of the logarithm of their biclique partition number. Equivalently, for infinitely many $n$ there is an $n$-vertex graph $H$ such that the clique versus independent-set problem on $H$ requires $\widetilde\Omega(\log^2 n)$ bits of conondeterministic communication. The polyloglog factors hidden by the tilde are not computed in~\cite{bbgjk23}. In this section we trace the quantitative bounds stated in the cited sources, together with an explicit padding step, and show that they yield total Boolean matrices $M$ satisfying
\begin{equation}
 D(M)
 \ge
 \log p(M)
 =
 \Omega(\frac{(\log \rank_{\R}(M))^2}
 {(\log\log \rank_{\R}(M))^{7}}).
 \label{eq:bbgjk_seven}
\end{equation}
This justifies the exponent $7$ in Table~\ref{tab:log_rank_bounds}, which compares with the exponent $2$ of Eq.~\eqref{eq:current_lower} and the exponent $1$ of Theorem~\ref{thm:route_two_main}-\ref{part:route_two_final}.

\paragraph{Measures.}
Following~\cite{bbgjk23}, for a total Boolean function $g$ on $N$ variables, $\mathsf{C}_1(g)$ is the least $k$ such that $g$ can be written as a $k$-DNF, $\mathsf{C}_0(g)$ is the least $k$ such that $g$ can be written as a $k$-CNF, and $\mathsf{UC}_1(g)$ is the least $k$ such that $g$ can be written as an unambiguous $k$-DNF; thus $\mathsf{UC}_1$ coincides with the measure $\mathsf{UP}^{\mathsf{dt}}$ of Section~\ref{sec:lower_proof}. We write $\mathsf{C}(g):=\max\{\mathsf{C}_0(g),\mathsf{C}_1(g)\}$. For a partial function $f\colon\{0,1\}^{N}\to\{0,1,*\}$ and an input $x$, the measure $\mathsf{C}_{\overline1}(f,x)$ is the least size of a partial assignment consistent with $x$ all of whose extensions $x'$ satisfy $f(x')\in\{0,*\}$, and $\mathsf{C}_{\overline0}(f,x)$ is defined symmetrically with $\{1,*\}$.

For a total two-party problem $F$, let $\operatorname{cov}_0(F)$ be the minimum number of $0$-monochromatic rectangles whose union is $F^{-1}(0)$, and define the conondeterministic communication complexity $\mathsf{coNP}^{\mathsf{cc}}(F):=\lceil\log\operatorname{cov}_0(F)\rceil$. As in Section~\ref{sec:proof_overview}, $\mathsf{UP}^{\mathsf{cc}}(F):=\lceil\log\chi_1(F)\rceil$ denotes the unambiguous communication complexity. For an $n$-vertex graph $H$, the clique versus independent-set problem $\operatorname{CIS}_H$ gives Alice a clique $x$ and Bob an independent set $y$, both as vertex sets, with output $|x\cap y|\in\{0,1\}$; a clique and an independent set share at most one vertex, so the output is well defined.

\paragraph{Ingredients.}
The chain of~\cite{bbgjk23} consists of the following four steps, quoted with their exact parameters.

\begin{lemma}[{EAH function; \cite[Section~3]{bbgjk23}}]
\label{lem:bbgjk_eah}
For every sufficiently large $m$ there is a partial function $\operatorname{Eah}_m$ on $N_0:=2m^2$ variables and an input $z\in\operatorname{Eah}_m^{-1}(*)$ such that
\[
 \mathsf{C}(\operatorname{Eah}_m)\le 200\,m\log m,
 \qquad
 \mathsf{C}_{\overline0}(\operatorname{Eah}_m,z)\ge m^2/100,
 \qquad
 \mathsf{C}_{\overline1}(\operatorname{Eah}_m,z)=m^2.
\]
\end{lemma}

\begin{lemma}[{Cheat sheets; \cite[Section~5.1]{bbgjk23}}]
\label{lem:bbgjk_cs}
Let $f$ be a partial function on $N$ variables and $x\in f^{-1}(*)$, and suppose $\min\{\mathsf{C}_{\overline0}(f,x),\mathsf{C}_{\overline1}(f,x)\}<N$. Then there is a total function $g$ on $n_g\le 3N^2\log^2N$ variables with
\[
 \mathsf{C}_0(g)\ge\min\{\mathsf{C}_{\overline0}(f,x),\mathsf{C}_{\overline1}(f,x)\}
 \qquad\text{and}\qquad
 \mathsf{UC}_1(g)\le 3\,\mathsf{C}(f)\log^2N.
\]
\end{lemma}

\begin{lemma}[{Lifting; \cite[Theorem~4 and Eq.~(2)]{g15}, due to~\cite{glmwz16}}]
\label{lem:bbgjk_lift}
There is an absolute constant $\beta_0>0$ such that the following holds for the inner-product gadget $\operatorname{IP}_b$ on $b:=100\lceil\log N\rceil$ bits per party. For every total function $g$ on $N$ variables, the problem $F:=g\circ\operatorname{IP}_b^{N}$ satisfies, for all sufficiently large $N$,
\[
 \mathsf{coNP}^{\mathsf{cc}}(F)\ge\beta_0\,\mathsf{C}_0(g)\cdot b
 \qquad\text{and}\qquad
 \mathsf{UP}^{\mathsf{cc}}(F)\le\mathsf{UC}_1(g)\cdot(\lceil\log N\rceil+2b).
\]
\end{lemma}

\begin{lemma}[{CIS completeness; \cite{y91}, see also~\cite[Figure~1]{g15}}]
\label{lem:bbgjk_cis}
For every total two-party problem $F$ there is a graph $H$ on $n:=2^{\mathsf{UP}^{\mathsf{cc}}(F)}$ vertices such that $F$ is a subproblem of $\operatorname{CIS}_H$ under local input maps; consequently $\mathsf{coNP}^{\mathsf{cc}}(\operatorname{CIS}_H)\ge\mathsf{coNP}^{\mathsf{cc}}(F)$. Moreover, $D(\operatorname{CIS}_H)=O(\log^2 n)$ for every $n$-vertex graph $H$.
\end{lemma}

\begin{proposition}[Explicit form of the BBGJK separation]
\label{prop:bbgjk_seven}
For all sufficiently large $m$, the chain of Lemmas~\ref{lem:bbgjk_eah}--\ref{lem:bbgjk_cis} produces an $n$-vertex graph $H$ with
\[
 \mathsf{coNP}^{\mathsf{cc}}(\operatorname{CIS}_H)
 =
 \Omega(\frac{\log^2 n}{(\log\log n)^{7}})
\]
and a total Boolean matrix $M$ with $\rank_{\R}(M)\le n+1$ satisfying Eq.~\eqref{eq:bbgjk_seven}.
\end{proposition}

\begin{proof}
All estimates below hold for all sufficiently large $m$, and all constants are absolute.

\emph{Step 1: the total function.}
Apply Lemma~\ref{lem:bbgjk_cs} to $f:=\operatorname{Eah}_m$ and $x:=z$ from Lemma~\ref{lem:bbgjk_eah}. The hypothesis holds because $\min\{\mathsf{C}_{\overline0},\mathsf{C}_{\overline1}\}\le m^2<2m^2=N_0$. Since $(2\log m+1)^2\le 5\log^2m$ for $m\ge 32$, the resulting total function $g$ satisfies
\[
 \mathsf{C}_0(g)\ge m^2/100,
 \qquad
 \mathsf{UC}_1(g)\le 3\cdot(200\,m\log m)\cdot 5\log^2 m=3000\,m\log^3m,
\]
Adjoining dummy input variables does not change $\mathsf{C}_0(g)$ or $\mathsf{UC}_1(g)$. We therefore pad $g$, if necessary, and retain the notation $g$ and $n_g$, so that
\[
 m^2\le n_g\le 3N_0^2\log^2N_0\le 60\,m^4\log^2m.
\]
Consequently,
\[
 2\log m\le\log n_g\le 5\log m.
\]

\emph{Step 2: lifting.}
Apply Lemma~\ref{lem:bbgjk_lift} to $g$ with $N:=n_g$, so that $200\log m\le b\le 600\log m$. The lifted problem $F$ satisfies
\[
 \mathsf{coNP}^{\mathsf{cc}}(F)
 \ge
 \beta_0\cdot\frac{m^2}{100}\cdot 200\log m
 =
 2\beta_0\,m^2\log m
\]
and, using $\lceil\log n_g\rceil+2b\le 1206\log m$,
\[
 L:=\mathsf{UP}^{\mathsf{cc}}(F)
 \le
 3000\,m\log^3m\cdot 1206\log m
 \le
 4\cdot10^{6}\,m\log^4m.
\]

\emph{Step 3: the graph.}
Apply Lemma~\ref{lem:bbgjk_cis} to $F$, obtaining $H$ on $n=2^{L}$ vertices with
\[
 \mathsf{coNP}^{\mathsf{cc}}(\operatorname{CIS}_H)
 \ge
 \mathsf{coNP}^{\mathsf{cc}}(F)
 \ge
 2\beta_0\,m^2\log m.
\]
We claim $m\le L\le m^2$. The upper bound is immediate from Step~2 since $4\cdot10^{6}\log^4m\le m$. For the lower bound, the final claim of Lemma~\ref{lem:bbgjk_cis} gives an absolute $\kappa$ with
$\mathsf{coNP}^{\mathsf{cc}}(\operatorname{CIS}_H)\le D(\operatorname{CIS}_H)\le\kappa L^2$, so $\kappa L^2\ge 2\beta_0 m^2\log m$ and hence $L\ge m$. Consequently
\[
 \log m\le\log L\le 2\log m.
\]
Inverting Step~2 with $\log^4m\le\log^4L$ gives $m\ge L/(4\cdot10^{6}\log^4L)$, and therefore
\[
 \mathsf{coNP}^{\mathsf{cc}}(\operatorname{CIS}_H)
 \ge
 2\beta_0\cdot\frac{L^2}{(4\cdot10^{6})^2\log^8L}\cdot\frac{\log L}{2}
 =
 \Omega(\frac{L^2}{\log^7L})
 =
 \Omega(\frac{\log^2 n}{(\log\log n)^{7}}).
\]

\emph{Step 4: the matrix.}
Index rows by the cliques and columns by the independent sets of $H$, and let $M[x,y]:=1$ if $x\cap y=\varnothing$ and $M[x,y]:=0$ otherwise. Since $|x\cap y|=\sum_{v\in V(H)}\mathbf{1}[v\in x]\,\mathbf{1}[v\in y]$, the matrix of the values $|x\cap y|$ is a sum of $n$ rank-one matrices, and $M$ is the all-ones matrix minus it; hence $r:=\rank_{\R}(M)\le n+1$ and $\log r\le L+1$. The $1$-entries of $M$ are exactly the $0$-inputs of $\operatorname{CIS}_H$, so a partition of them into $p(M)$ disjoint $1$-monochromatic rectangles is in particular a $0$-cover of $\operatorname{CIS}_H$, giving
\[
 \log p(M)\ge\mathsf{coNP}^{\mathsf{cc}}(\operatorname{CIS}_H)-1.
\]
The $1$-leaves of an optimal deterministic protocol for $M$ partition the $1$-entries into at most $2^{D(M)}$ disjoint $1$-monochromatic rectangles, so $D(M)\ge\log p(M)$. Finally, the elementary bound of Section~\ref{sec:introduction} gives $D(M)\le r+1$, so $r\to\infty$, and the function $\psi(t):=t^2/(\log t)^7$ is increasing for all sufficiently large $t$. Since $\log r\le L+1$,
\[
 D(M)
 \ge
 \log p(M)
 \ge
 \Omega(\psi(L))
 \ge
 \Omega(\psi(\log r))
 =
 \Omega(\frac{(\log r)^2}{(\log\log r)^{7}}),
\]
which is Eq.~\eqref{eq:bbgjk_seven}.
\end{proof}

\paragraph{The ledger.}
Writing every logarithmic loss in units of $\log m=\Theta(\log\log n)$, and noting that a factor on the $\mathsf{UP}$ side enters $L=\log n$ and is therefore squared in the final ratio while a factor multiplying the $\mathsf{coNP}$ lower bound counts once negatively, the exponent $7$ decomposes as follows. The EAH estimate $\mathsf{C}(\operatorname{Eah}_m)=O(m\log m)$ contributes $+2$. The factor $\log^2N$ in the cheat-sheet estimate for $\mathsf{UC}_1(g)$ contributes $+4$. The factor $\lceil\log N\rceil+2b$ in the lifting upper bound contributes $+2$, whereas the factor $b$ in the lifting lower bound contributes $-1$. The embedding into $\operatorname{CIS}$ and the passage to the matrix introduce no further logarithmic loss. In total, $2+4+2-1=7$.

This ledger records only the losses in the displayed bounds used above; it does not assert that the exponent $7$ is optimal for the BBGJK framework. If the factor $b$ in the lower bound of Lemma~\ref{lem:bbgjk_lift} is omitted, the same calculation gives the weaker exponent $8$. By contrast, the Booleanized GPW route gives the exponent $2$ in Eq.~\eqref{eq:current_lower}, and the alphabet-preserving lift of Theorem~\ref{thm:route_two_main} gives the exponent $1$.

\end{document}